%% file: ec163-babaioff.tex
\documentclass[prodmode,acmec]{ec12-acmsmall}

\acmYear{2012}
\acmMonth{06}

\usepackage{amssymb}
\usepackage{amsmath}

\newtheorem{claim}[theorem]{Claim}

\title{On Bitcoin and Red Balloons}
\author{
  Moshe Babaioff \affil{Microsoft Research, Silicon Valley}
  Shahar Dobzinski \affil{Cornell University}
  Sigal Oren \affil{Cornell University}
  Aviv Zohar \affil{Microsoft Research, Silicon Valley}}

\begin{document}

%\setcounter{page}{0}
%\maketitle
%\thispagestyle{empty}




\begin{bottomstuff}
The second and third authors were partially supported by NSF grants IIS-0910664 and CCF-0910940. Some of the work was done while the second and third author we visiting Microsoft Research, Silicon Valley.

Addresses: Moshe Babaioff, Microsoft Research, Silicon Valley (moshe@microsoft.com) {and}
Shahar Dobzinski, Computer Science Department, Cornell University (shahar@cs.cornell.edu);
Sigal Oren, Computer Science Department, Cornell University (sigal@cs.cornell.edu);
Aviv Zohar, Microsoft Research, Silicon Valley (avivz@microsoft.com).
\end{bottomstuff}

\begin{abstract}
Many large decentralized systems rely on information propagation to ensure
their proper function.
% However, it is common that participants that are aware of the information compete for some reward, and thus have an incentive {\em not} to propagate information.
We examine a common scenario in which only participants that are aware of the
information can compete for some reward, and thus informed participants have
an incentive {\em not} to propagate information to others. One recent example
in which such tension arises is the 2009 DARPA Network Challenge (finding red
balloons). We focus on another prominent example: Bitcoin, a decentralized
electronic currency system.

Bitcoin represents a radical new approach to monetary systems. It has been
getting a large amount of public attention over the last year, both in policy
discussions and in the popular press \cite{NY11,technology-review}. Its
cryptographic fundamentals have largely held up even as its usage has become
increasingly widespread. We find, however, that it exhibits a fundamental
problem of a different nature, based on how its incentives are structured. We
propose a modification to the protocol that can eliminate this problem.

Bitcoin relies on a peer-to-peer network to track transactions that are
performed with the currency. For this purpose, every transaction a node
learns about should be transmitted to its neighbors in the network. As the
protocol is currently defined and implemented, it does not provide an
incentive for nodes to broadcast transactions they are aware of. In fact, it
provides an incentive not to do so. Our solution is to augment the protocol
with a scheme that rewards information propagation. Since clones are easy to
create in the Bitcoin system, an important feature of our scheme is
Sybil-proofness.

We show that our proposed scheme succeeds in setting the correct incentives, that it is
Sybil-proof, and that it requires only a small payment overhead, all this is
achieved with iterated elimination of dominated strategies. We complement this result by showing that there are no reward schemes in which information propagation and no self-cloning is a dominant strategy.
\end{abstract}

\maketitle

\input{intro}
\input{almost-uniform}

\input{impossibility}
\subsubsection*{Acknowledgements} We thank Jon Kleinberg and Ilya Mironov for helpful discussions and valuable comments.

\bibliographystyle{acmsmall}

\bibliography{bitcoin}

\appendix
\section*{APPENDIX}
\input{overview}
\end{document}

%% file: intro.tex
\section{Introduction}

In 2009 DARPA announced the DARPA Network Challenge, in which participants
competed to find ten red weather balloons that were dispersed across the
United States~\cite{red_balloon}. Faced with the daunting task of locating
balloons spread across a wide geographical area, participating teams
attempted to recruit individuals from across the country to help. The winning
team from MIT~\cite{PPRCCMP11}, incentivized balloon hunters to seek balloons
by offering them rewards of $\$2000$ per balloon. Furthermore, after
recognizing that notifying individuals from all over the US about these
rewards is itself a difficult undertaking, the MIT team cleverly offered
additional rewards of $\$1000$ to the person who directly recruited a balloon
finder, a reward of $\$500$ to his recruiter, and so on. These additional
payments created the incentive for participants to spread the word about
MIT's offer of rewards and were instrumental in the team's success. In fact,
some additional rewards are necessary: each additional balloon hunter
competes with the participants in his vicinity, and reduces their chances of
getting the prize.

% Moshe 2.2: changes below. Now we know that MIT asked for SSN, so we cannot really claim MIT is open to the problem of Sybil attacks. I have replaced the footnote by another.
% OLD:
% MIT's scheme still requires further improvement.
The MIT scheme as described above is susceptible to the following attack.
A participant can
create a fake identity, invite the fake identity to participate, and use that
identity to recruit others.
%\footnote{In the actual implementation the MIT team has required participants to disclose their Social Security Number to receive payment. Such a solution has its costs (e.g. some people might not be able or refuse to participate) and is not feasible in every information propagation scenario.} By doing so he increases his prize by 50\%.
% \footnote{Indeed, we have no evidence of such attacks in the DARPA challenge. If no such attacks were made, one possible explanation is the short time span of the challenge and its non-commercial, scientific essence. It seems quite plausible that if the challenge is repeated several times such attacks will become common.}
Thus,
when participants can create a fake identities
the reward scheme should be carefully
designed so it does not create an incentive for such attacks. Our goal
is to design reward schemes that incentivize \emph{information propagation}
and counter the dis-incentive that arises from the competition with other
nodes, and are \emph{Sybil proof} (robust against creating clones, or Sybil attacks) while having
a \emph{low overhead} (a total reward that is not too high).

A related setting is a raffle, in which people purchase numbered tickets in
hopes of winning some luxurious prize. Each ticket has the same probability
of winning, and the prize is always allocated. As more tickets are sold, the
probability that a certain ticket will win decreases. In this case again,
there is a clear tension between the organizer of the raffle, who wants as
many people to find out about the raffle, and the participants who have
already purchased tickets and want to increase their individual chances of
winning. The lesson here is simple, to make raffles more successful
participants should be incentivized to spread the word. One example of a
raffle which is already implemented in this way is Expedia's ``FriendTrips'', in which the
more friends you recruit the bigger your probability of winning.

%We identify a potential incentives problem which might disrupt information
%propagation in the Bitcoin protocol, and propose a fix. This problem is not
%unique to Bitcoin, and arises in many other a-priori unrelated settings as
%well. To explain our setting and the problem in the Bitcoin protocol, we
%start with a simplified high level description of it (see Section
%\ref{sec-bitcoin} for a more detailed description).

As apparent from the previous examples, the tension between information
propagation and an increased competition is a general problem. We identify an
instantiation of this tension in the Bitcoin protocol, our main example for
the rest of the paper. Bitcoin is a decentralized electronic currency system
proposed by Satoshi Nakamoto in $2008$ as an alternative to current
government-backed currencies.\footnote{The real identity of Satoshi Nakamoto
remains a mystery. A recent \emph{New Yorker} article \cite{NY11}
%avivz: keep this citation, or just mention the article in the text?
attempts to identify him.} Bitcoin has been actively running since $2009$.
Its appeal lies in the ability to quickly transfer money over the internet,
and in its relatively low transaction fees.\footnote{There are additional
properties that some consider as benefits: Bitcoins are not controlled by any
government, and its supply will eventually be fixed. Additionally, it offers
some degree of anonymity (although this fact has been
contested~\cite{Reid11}).} As of February 2012, it has 8.2 million coins in
circulation (called \emph{Bitcoins}) which are traded at a value of
approximately 6 USD per bitcoin. Below, we give a brief explanation of the
Bitcoin protocol and then explain where the incentives problem appears. A more comprehensive description of the protocol is given in
Appendix~\ref{sec-bitcoin}.

Bitcoin relies on a peer-to-peer network to verify and authorize all
transactions that are performed with the currency. Suppose that Alice wants
to pay a sum of $30$ bitcoins for her stay in a hotel. Alice
cryptographically signs a transaction to transfer $30$ bitcoins from her to
the hotel, and sends the signed transaction to a small number of nodes in the
network. Each node in the network propagates the transaction to its
neighbors. A node that receives the transaction verifies that Alice has
signed it and that she does indeed own the bitcoins she is attempting to
transfer. The node then tries to ``authorize'' the transaction
% \footnote{In fact, a node authorizes a set of several transactions together (a ``block''). To keep the presentation simple, this simplified description considers only a single transaction. Our proposed solution applies to the case of multiple transactions as well.}
by attempting
to solve a computationally hard problem (basically inverting a hash
function). Once a node successfully authorizes a transaction, it sends the
``proof'' (the inverted hash) to all of its neighbors. They in turn, send the
information to all of their neighbors and so on. Finally, all nodes in the
network ``agree'' that Alice's bitcoins have been transferred to the hotel.

To motivate them to authorize transactions, nodes are offered a payment in
bitcoins for successful authorizations. The system is currently in its
initial stages, in which nodes are paid a predetermined amount of bitcoins
that are created ``out of thin air''. This also slowly builds up the supply
of bitcoins. But Bitcoin's protocol specifies an exponentially decreasing
rate of money creation that effectively sets a cap on the total number of
bitcoins that will be in circulation. As this payment to nodes is slowly
phased out, bitcoin owners that want their transactions approved are supposed
to pay fees to the authorizing nodes.

This is where the incentive problem manifests itself. A node in the network
has an incentive to keep the knowledge of any transaction that offers a fee
to itself, as any other node that becomes aware of the transaction will
compete to authorize first and claim the associated fee. For example, if only
a single node is aware of a transaction, it can eliminate competition
altogether by not distributing information further. With no competition, the
informed node will eventually succeed in authorizing and collect the fee. The
consequences of such behavior may be devastating: as only a single node in
the network works to authorize each transaction, authorization is expected to
take a very long time.\footnote{Bitcoin's difficulty level is adjusting
automatically to account for the total amount of computation in the network.
The expected time of a single machine to authorize a transaction is currently
on the order of $60$ days, instead of the $10$ minutes it is expected to take
if the entire network competes to authorize.}

% Moshe 4.2: many changes below
We stress that every change in the Bitcoin protocol {\em has} to take into
account Sybil attacks, as false identities are a prominent concern. Given
this concern the Bitcoin protocol
%was built under the assumption that nodes can create false identities, and
is designed to ensure that if a majority of {\em processing power}, rather than majority of declared identities, follow the protocol, it will not be possible to manipulate the history of authorized transactions.
% In fact, the Bitcoin protocol is built around the assumption that nodes can create false identities, and considers a transaction fully approved only once nodes that control a majority of the CPU power in the network have accepted it, rather than just a majority of the nodes.
Notice that simply using a majority of declared identities is vulnerable to Sybil attacks, as identities are easy to spoof.
Therefore any reward scheme for transaction propagation must also discourage Sybil attacks.

\subsubsection*{A Simplified Model}

We now present our model, using Bitcoin as our running example. For
simplicity assume that only one transaction is awaiting authorization. We
model the authorization process as divided into two phases: the first phase
is a distribution phase. The second one is a computation phase, in which
every node that has received the transaction is attempting to authorize it.

We start with describing the distribution network. We assume that the network consists of a forest of $d$-ary directed
trees, each of them of height $H$. The \emph{distribution phase} starts when the buyer sends the details of the transaction to
the $t$ roots of the trees (which we shall term \emph{seeds}). We think of
the trees as directed, since the information (about the transaction) flows
from the root towards the leaves. The {\em depth} of a node is the number of nodes on the path from the root of the node's tree (the seed) to the node. Let $n=t\cdot \frac {d^{H}-1} {d-1}$ be
the total number of nodes.

%In general, it is common to think of an efficient peer to peer network as a
%random graph, in which messages propagate quickly, multiplying the number of
%recipients as the message is sent to each additional layer. Since we are not
%able to solve for the general case of random graphs, we will further simplify
%that and assume that the network consists of a forest of $d$-ary directed
%trees, each of them of height $H$.\footnote{The intuition for this
%simplification is that when $d$ is small relatively to the number of nodes
%$n$, message propagation in a random graph resembles a tree. In some sense,
%% Moshe 2.2:
%% OLD: the case of trees is harder than general graphs as each node monopolizes the flow of information to its descendants.
%creating the right incentives is harder for trees than for general graphs, as each node in a tree monopolizes the flow of information to its descendants.}

In the distribution phase, each node $v$ can send the transaction to any of
its neighbors after adding any number of fake identities of itself before
sending to that neighbor.
% All $v$'s fake identities are connected to the same set of children.
Node $v$ can only relay the transaction to its original children (adding fake identities does not change the neighbors of $v$).
Naturally, a node can condition its behavior only on the information available to it: the \emph{length}
of the chain above it (which can possibly include false identities that were
produced by its ancestors).

% Moshe 4.2: major changes below to be more precise about nodes vs. identities:
In the \emph{computation phase} each node that is aware of the transaction tries to authorize it\footnote{In the Bitcoin protocol multiple transactions are placed at the same ''block'' and the cost of authorizing the entire block is fixed, essentially independent of the content of the block. Focusing on the issue of incentives to propagate information, in this paper we have
%neglected the cost of computation and have
assumed that every node tries to authorize any transaction it is aware of
(since the cost of adding it the block it is attempting to authorize anyway
is negligible). In future work one might consider studying the issue of
sharing the cost of computation between transactions in a block in a way that
motivates the nodes to compute.}. If there are $k$ such nodes, each of them
has the same probability of $\frac 1 k$ to be the first to authorize the
transaction. Note that $k$ is the number of real nodes, fake identities do
not increase the probability of winning the reward. A node $v$ that
authorizes the transaction can declare that it did so with any identity $p_h$
it created. This determines a chain of identities $p_1,\ldots, p_h$ which we
call the {\em authorizing chain}. The chain ends with $p_h$, the declared
identity of the authorizer, and starts with $p_1$, an identity of the seed
that roots the tree $v$ is in. This chain is a superset of the real path in
the tree, potentially including clones of some nodes (nodes are not able to
remove predecessor identities from the chain due to the use of cryptographic
signatures). Rewards are allocated to identities on the authorizing chain in
reverse order, reward $r_1$ is allocated to the authorizing identity and
$r_h$ to the seed. Each node's reward is the sum of the rewards of all its
identities (true and fake) in the authorizing chain. Currently the bitcoin
protocol requires to attach a minimal reward $c>0$ to every transaction. This
reward is given to the node that authorizes the transaction. This is the
prize that the nodes are competing on. We normalize $c$ to be $1$.

%We assume that there exists a
%minimal payment $c$ that at least covers the expected computation cost for a
%single node to authorize the transaction {\em by itself} such that any node that is
%promised a reward of $c$ will attempt to authorize the transaction.\footnote{Assume that $n$ nodes are trying to authorize the transaction. The total expected cost of {\em all} of the nodes till one of them authorizes the transaction is $c$ (the same as the cost of one node authorizing the transaction alone), so each node has expected cost of $c/n$ till one node authorizes the transaction. Yet, we pay the authorizing node $c$ to compensate for the fact that he only authorizes the transaction with probability $1/n$, so his expected payment is exactly his cost.}
%Thus, we
%require that the authorizer reward $r_1$ is at least $c$.\footnote{In the
%actual Bitcoin protocol many transactions are authorized together, and thus
%the minimum reward $r_1$ per transaction is relatively small since the
%computation cost is split between many transactions.} For a smaller reward,
%nodes will refuse to participate in the first place (due to individual
%rationality). We normalize $c$ to be $1$.

We assume that all players are expectation maximizers, so the utility of a player in a given profile is its expected utility, where expectation is taken over the random selection of the authorizer $w$.

\subsubsection*{Our Results}

A successful reward scheme has to posses several properties. First, it has to incentivize \emph{information propagation} and \emph{no duplication}. That is, it will be in a node's best interest to distribute the transaction to all its children without duplicating itself,
as well as never duplicating when it authorizes. Second, at the end of the distribution phase most of the nodes have to be aware of the transaction. Lastly, we would like to achieve this with small rewards, while minimizing the number of seeds, to ease the burden of initial distribution.

Designing a scheme that simultaneously has all these properties is obviously
not a trivial task, and in particular specific care has to be given to the
Sybil proofness requirement. The current literature on the Sybil proof
mechanisms contains mostly negative results, and in the few cases in which a
positive result is obtained, the setting is usually very specific. The reason is previous works insisted that
creating Sybils will \emph{never} be profitable. In contrast, we use our
game theoretic model to take an alternative path: we will guarantee these
properties, Sybil proofness in particular, \emph{in equilibrium}.

Using this concept, we introduce a new family of schemes: {\em almost uniform} schemes. Each member in the family is parameterized by a height parameter $\mathcal H$ and a reward parameter $\beta$. Let $v$ be the node that authorized the transaction and suppose that $v$ is
the $l$'th node in the chain. If $l>\mathcal H$ no node is rewarded (so
nodes ``far'' from the seed will not attempt to authorize the transaction). Otherwise, each node in the chain except $v$ gets a reward of
$\beta$, and $v$ gets a reward of $1+(\mathcal H-l+1)\beta$. We show that if there are $\Omega (\beta^{-1})$ seeds, only strategy profiles that exhibit information propagation and no duplication survive every order of iterated removal of dominated strategies. Furthermore, there exists an order in which no other strategy profiles survive.

Iterated removal of dominated strategies is the following common technique
for solving games: first the set of surviving strategies of each player is
initiated to all its strategies. Then at each step, a strategy of one of the
players is eliminated from the set. The strategy that is eliminated is one
that is dominated by some other strategy of the same player (with respect to
the strategies in the surviving sets of all other players). This process
continues until there is no strategy that can be removed. The solution
concept prescribes that each player will only play some strategy from his
surviving set. In general the surviving sets can depend on the order in which
strategies are eliminated.
Yet, we show that in our case, regardless of the order of elimination,
profiles of strategies in which every node propagates information and never duplicates, survive.
Moreover, for a specific order of elimination they are the only strategies that survive.

Although the description of our schemes is simple, the analysis itself is quite delicate, but let us give a bird-eye view of it now. The intuition behind the elimination process is that decreasing competition
by your own descendants might be unprofitable due to
% Moshe 4.2: we did not define "distribution rewards" and it was hard to understand, so I made this more explicit.
%OLD: losing possible distribution rewards
possibly losing rewards for relaying, in case one of these descendants authorizes.
Consider one particular node. It faces the following dilemma: by duplicating itself one more time it increases its reward in case one of its descendants authorizes the transaction, but it (potentially) decreases the number of descendants aware of the transaction, and thus may decrease its expected reward. To see this, consider, for simplicity, a seed. Suppose, for the sake of explanation, that all of its descendants always propagate information and never duplicate. If the seed does not clone itself then each one of its children spans a $d$-array tree of height $\mathcal H-1$. If it clones itself once, then each one of its children spans a $d$-array tree of smaller height of $\mathcal H-2$, and the number of descendants of the node that receive the information decreases by almost factor $d$.
This simple observation stands at the heart of the analysis, but the reader should note that it is over simplified: it is far from being clear that the descendants will propagate information with no duplication in this scheme.

As explained, if the amount of external competition
from non-descendent nodes is small, a node prefers not to distribute the
transaction, and thus increase the probability that it receives the reward
for authorization. However, if sufficiently many
non-descendent nodes are aware of the transaction, the node prefers to
duplicate itself one less time, and thus distribute the transaction to its
children and increase its potential distribution reward. Once all nodes
increase distribution, the arms race begins: the competition each node faces
is greater, and again it prefers to duplicate itself one less time and
distribute all the way to its grand-children. As this process continues, all
nodes eventually prefer to distribute fully and never to duplicate.

Two values of $\beta$ give us schemes that are of particular interest. The first is when $\beta=1$. In this case the $(1,\mathcal
H)$-almost-uniform scheme requires only a constant number of seeds and the
total payment is always $O(\mathcal H)$. The second scheme is the $(\frac 1 {\mathcal H},\mathcal H)$-almost-uniform scheme.
This scheme works if the number of seeds is
$\Omega(\mathcal H)$. Its total payment is $2$.

We combine both schemes to create a reward scheme that has \emph{both} a
constant number of seeds and a constant overhead. The \emph{Hybrid Scheme}
first distributes the transaction to a constant number of seeds using the
\emph{$(1,1+ \log_d H)$-almost-uniform scheme}. Then we can argue that at
least $H$ nodes are aware of the transaction. This fact enables us to further
argue that the $(\frac 1  H, H)$-almost uniform scheme guarantees that
the transaction is distributed to trees of height $H$. At the end of the
distribution phase, most of the nodes that are aware of the transaction
receive a reward of $\frac 1 H$ if they are in the successful chain, so the
expected overhead is low. We have the following (imprecisely stated) theorem:

\vspace{0.1in}\noindent {\bf Theorem:} In the hybrid rewarding scheme, if the
number of seeds $t \geq 14$, the only strategies that always survive iterated elimination of dominated strategies exhibit information propagation and no duplication. In addition, there exists an elimination order in which the only strategies that survive exhibit information propagation and no duplication.
Furthermore, the expected total sum of payments is at most $3$.

\vspace{0.1in}\noindent Notice that this scheme exhibits in equilibrium low
overhead, Sybil proofness, and provides the nodes with an incentive to propagate
information.

Iterated removal of dominated strategies is a strong solution concept,
but ideally we would like our rewarding scheme to achieve all
desired properties in the stronger notion of dominant strategies equilibrium.
However, we show that this is impossible to achieve:

\vspace{0.1in}\noindent {\bf Theorem:} Suppose that $ H\geq 3$. There is no
Sybil-proof reward scheme in which information propagation and no duplication
are dominant strategy for all nodes at depth $3$ or less.

\subsubsection*{Related Work}
The paper describing Bitcoin's principles was originally published as a white
paper by Satoshi Nakamoto~\citeyear{Satoshi}. The protocol was since developed in
an open-source project. To date, no formal document describes the protocol in
its entirety, although the open source community maintains wiki pages devoted
to that purpose\nocite{wiki}. Some research has been conducted with
regards to its privacy~\cite{Reid11}. Krugman~\citeyear{K} discusses some related economic concerns.

Chitnis et al. \citeyear{CHKM12} discuss other game theoretic aspects of the DARPA Network Challenge -- such as finding the best strategy for a group to find the balloons.

The subject of incentives for information dissemination, especially in the
context of social and peer-to-peer networks has received some attention in
recent years. Kleinberg and Raghavan~\citeyear{Kleinberg05} consider
incentivizing nodes to search for information in a social network. A node
that possesses the information is rewarded for relaying it back, as are nodes
along the path to it. A key difference between their model and ours, is that
their model assumes nodes either posses the sought-after information or do
not. If they do, then there is no need for further propagation, and if they
do not, they cannot themselves generate the information, and so do not
compete with nodes they transmit to (they do however compete for the
propagation rewards with other unrelated nodes).

Douceur and Moscibroda \citeyear{DM07} design recruitment mechanisms similar
in spirit to our own to aid incremental deployment in networked systems. They
formalize various properties schemes should ideally uphold and present
mechanisms as well as impossibility results. Unlike our work, their model and
properties are not game theoretic, and the strategies of different
participants are not considered or analyzed with respect to any solution
concept.

Emek \emph{et. al.}~\citeyear{Emek11} consider reward schemes for social
advertising scenarios. In their model, the goal of the scheme is to
advertise to as many nodes as possible, while rewarding nodes for forwarding
the advertisement. Unlike their scenario, in our case the scheme is
eventually only aware of a chain that results in a successful authorization,
and only awards nodes on the path to the successful authorizer. Drucker and Fleischer ~\citeyear{Drucker12}
propose sybil-proof mechanisms for a similar setting.

Other works consider propagation in social networks without the aspect of
incentives. For example~\cite{Dyagilev10}, which considers ways to detect
propagation events and examines data from cellular call data.

\subsubsection*{Future Research}

In this work we propose a novel low cost reward scheme that incentivizes
information propagation and is Sybil proof. Currently we
assume that the distributing network is modeled as a forest of $t$ complete $d$-ary trees. A
challenging open question is to consider other types of networks, for example random $d$-regular graphs, which may capture better peer-to-peer networks\footnote{Notice that when $d$ is small (e.g. constant) and the number of nodes
$n$ is large, transaction propagation in a random graph (till about $1/d$ fraction of the graph receives the information) resembles a tree. In some sense, building incentives for information propagation in trees is harder than general graphs as each node monopolizes the
flow of information to its descendants.}.
Another interesting extension to consider
is one in which nodes have different processing power. That is,
each node $v$ has a processing power $CPU_v$, then the probability of a node $v$ that is aware of the transaction to authorize
it is $\frac {CPU_v} {\Sigma_uCPU_u}$.

We leave the analysis of these models, as well as the implementation and empirical experimentation, to future research.

%% file: almost-uniform.tex
\section{New Reward Schemes}

Our main result is the Hybrid scheme, a reward scheme that requires only a constant number of seeds and a constant overhead. We will show that the only strategies that always survive iterated removal of dominated strategies are ones that exhibit information propagation and no duplication. The basic building blocks for this construction is a family of schemes, almost-uniform schemes, with less attractive properties. However, we will show that combining together two almost-uniform schemes enables us to obtain the improved properties of the Hybrid scheme.

\subsection{The Basic Building Blocks: Almost-Uniform Schemes}
The $(\beta,\mathcal H)$-almost-uniform scheme pays the authorizing node in a
chain of length $h$ a reward of $1+\beta \cdot (\mathcal H-h+1)$. The rest of
the nodes in the chain get a reward of $\beta$. If the chain length is
greater than $\mathcal H$ no node receives a reward. The idea behind giving
the authorizing node a reward of $1+\beta \cdot (\mathcal H-h+1)$ is to mimic
the reward that the node would have gotten if it duplicated itself $\mathcal
H-h$ times. %Therefore, we can assume that $p^{l,v}=0$ for every node $v$ and $l$.
Therefore, we can assume that no node duplicates itself before trying to
authorize the transaction, and focus only on duplication before sending to
its children. Figure~\ref{fig_uniform} illustrates the reward scheme.

\begin{figure}[!ht]
  \begin{center}
 \includegraphics[width=0.7\linewidth]{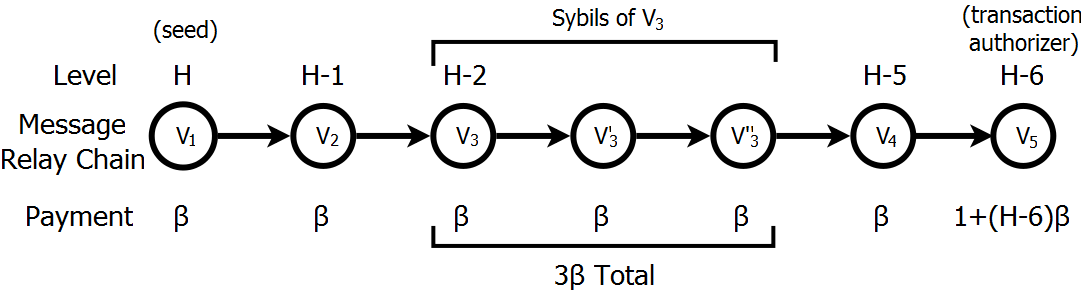}
    
  \end{center}
  \caption{\label{fig_uniform}
  An example of the transaction relay chain from the seed $v_1$ to the authorizing node $v_5$ in the $(\beta,\mathcal H)$-almost-uniform scheme. In the example, node $v_3$ added $2$ fake identities ($v'_3$ and $v''_3$) before relaying the transaction to $v_4$, all other nodes relay with no duplications. The payment to each identity appears below the chain, note that node $v_3$ receives $3\beta$ in total due to its clones. Each true identity is associated with a ''level'', representing the maximal number of identities (true and fake) it could have added to the prefix of the chain when receiving it, to get to a chain of length $H$.}
\end{figure}

\begin{theorem}\label{thm-almost-uniform}
Let $d\geq 3$. Suppose that the number of seeds is $t\geq 7$ and in addition there are initially at least $2\beta^{-1}+ 6$ nodes except the $t$ seeds that are aware of the transaction.
% Moshe 3.2 : I have change the statement here and in the other places below.
In the $(\beta,\mathcal H)$-almost-uniform scheme only profiles of strategies in which every node of depth at most $\mathcal H$ never duplicates and always propagates information survive in {\em every} iterated removal of weakly dominated strategies.
Furthermore, there is an elimination order in which {\em only} these profiles of strategies survive. The total payment is $1+\mathcal H\cdot \beta$.
\end{theorem}

While the $2\beta^{-1}+ 6$ additional nodes that are aware the transaction are not necessarily seeds, for now, one can think about them as additional seeds. The Hybrid scheme will exploit this extra flexibility to simultaneously obtain low overhead and small number of seeds. Before proving the theorem, let us mention two of its applications:

%In many games iterated removal of dominated strategies may lead to different outcomes. We show that in our case, every order of elimination leads to the same outcome:
%
%\begin{theorem}[Uniqueness]\label{thm-almost-uniform-unique}
%Let $d\geq 3$. Suppose that the number of seeds is $t\geq 7$ and in addition there are initially at least $2\beta^{-1}+ 6$ nodes except the $t$ seeds that know about the transaction. Then, in every elimination order of weakly dominated strategies propagating information and no duplication are the only surviving strategies.
%\end{theorem}
% Moshe 3.2 : I have change the statement of the corollary
\begin{corollary}
Let $d\geq 3$.
\begin{enumerate}
\item Consider the $(\frac 1 {\mathcal H}, \mathcal H)$-almost uniform scheme ($\beta=\frac 1 {\mathcal H}$),
with at least $2\mathcal H + 13$ seeds.
Only profiles of strategies in which every node of depth at most $\mathcal H$
never duplicates and always propagates information survive in {\em every} iterated removal of weakly dominated strategies.
 % in the $(\frac 1 {\mathcal H},\mathcal H)$-almost-uniform scheme.
Furthermore, there is an elimination order in which {\em only} these profiles of strategies survive.
%Furthermore, there is an elimination order in which these strategies are the only ones that survive.
The total payment is $2$.
\item Consider the $(1,\mathcal H)$-almost uniform scheme ($\beta=1$) with at least $15$ seeds.
%Then, information propagation and no duplication are the only strategies that always survive iterated removal of weakly dominated strategies in the $(\beta,\mathcal H)$-almost-uniform scheme. Furthermore, there is an elimination order in which these strategies are the only ones that survive.
Only profiles of strategies in which every node of depth at most $\mathcal H$
never duplicates and always propagates information survive in {\em every} iterated removal of weakly dominated strategies.
 % in the $(1,\mathcal H)$-almost-uniform scheme.
Furthermore, there is an elimination order in which {\em only} these profiles of strategies survive.
The total payment is $\mathcal H + 1$.
\end{enumerate}
\end{corollary}

Both parts of the corollary are direct applications of the theorem, by making sure that the number of initial seeds is at least $2\beta^{-1}+ 13$. The next subsection introduces the Hybrid scheme that combines between two almost uniform schemes, one with $\beta=1$ and one with $\beta=\frac 1 {\mathcal H}$, and obtain a scheme that uses both a constant number of seeds and its total expected payment is a small constant that does not depend on $\mathcal H$.

The rest of the section is organized as follows. In Subsection \ref{subsec-almost-uniform-proof} we prove our Theorem \ref{thm-almost-uniform}. The next subsection presents the Hybrid scheme.
% Finally, Subsection
Appendix \ref{subsec-implementation} discusses how to implement the Hybrid scheme within the framework of the Bitcoin protocol.

\subsection{The Final Construction: The Hybrid Scheme}

We now present our main construction, the Hybrid scheme. The scheme works as follows: we run the $(\frac 1 H,H)$-almost uniform scheme with a set of $A$ seeds ($|A|=a$) and simultaneously run the $(1,1+\log H)$-almost uniform scheme\footnote{All logarithms in this paper have base $d$, so we write $\log H$ to denote $\log_d H$.} with a set of $B$ seeds ($|B|=b$).

% Moshe 3.2 : I have change the statement of the theorem
\begin{theorem}
Assume that $d\geq 3$, $t\ge 15$ and consider the Hybrid scheme with set $A$ of size $a=t-7$ and set $B$ of size $b=7$.
Let $Z$ be the set of nodes of depth at most $H$ in trees rooted by seeds in $A$, and $1+\log H$ for nodes rooted by seeds in $B$.

In the Hybrid scheme only profiles of strategies in which every node in $Z$ never duplicates and always propagates information
survive in {\em every} iterated removal of weakly dominated strategies.
Furthermore, there is an elimination order in which {\em only} these profiles of strategies survive.
The total payment is in expectation at most $3$, and $2+\log H$ in the worst case.
%At least $\frac a {a+b}$-fraction
At least $\frac {t-7} t$-fraction
of the network is aware of the transaction after the distribution phase.
\end{theorem}
%\begin{theorem}
%Assume that  $a\geq b\geq 7$ and consider the Hybrid scheme with set $A$ of size $a$ and set $B$ of size $b$.
%Let $Z$ be the set of nodes of depth at most $H$ in trees rooted by seeds in $A$, and $1+\log H$ for nodes rooted by seeds in $B$.
%
%In the Hybrid scheme only profiles of strategies in which every node in $Z$ never duplicates and always propagates information
%survive in {\em every} iterated removal of weakly dominated strategies.
%Furthermore, there is an elimination order in which {\em only} these profiles of strategies survive.
%The total payment is in expectation at most $3$, and $2+\log H$ in the worst case.
%At least $\frac a {a+b}$-fraction
%of the network is aware of the transaction after the distribution phase.
%\end{theorem}
\begin{proof}
The theorem is a consequence of Theorem \ref{thm-almost-uniform}: all nodes up to depth $\log H$ in trees rooted by nodes in $B$ will propagate information and will not duplicate themselves. Thus, when considering the nodes in $A$, there are at least $7\cdot \frac{d\cdot d^{\log H}-1}{d-1}\geq 7H\geq 2H+6$ nodes that are aware of the transaction, which are not part of the trees rooted by nodes in $A$. Thus we can apply again Theorem \ref{thm-almost-uniform} to claim each seed in $A$ roots a full tree of height $H$.

Notice that the worst case payment is $1+\log H$ (because of the use of the $(1,1+\log H)$-almost uniform scheme). The expected payment is $$\dfrac {a\cdot \frac{d^H-1}{d-1}\cdot 2 + b\cdot \frac{dH-1}{d-1}\cdot (1+\log H)} {a\cdot \frac{d^H-1}{d-1}+b\cdot \frac{dH-1}{d-1}} \leq 3.$$ As for the number of nodes that are aware of the transaction at the end of the distribution phase: all trees rooted by nodes in $A$ are fully aware of the transaction. These nodes are
% Moshe 2.2:
at least $\frac a {a+b}$-fraction
%OLD : a fraction of $\frac a {a+b}$
of the network (notice that we conservatively ignored nodes rooted by seeds in $B$ that are aware of the transaction).
\end{proof}

\subsection{Proof of Theorem \ref{thm-almost-uniform}}\label{subsec-almost-uniform-proof}
Given specific values of $\beta$ and $\mathcal H$, the $(\beta,\mathcal H)$-almost uniform scheme defines the utilities of the nodes. We next consider the {\em $(\beta,\mathcal H)$-game} which is the game the nodes are playing given these utilities and their strategies.

\subsubsection{Notation and Formalities for the $(\beta,\mathcal H)$-game}\label{sec-prelim}
In order to present the proof we need to establish some notations regarding the $(\beta,\mathcal H)$-game.

We start by discussing the strategy space of a node. Only nodes that are aware of the transaction can relay it and get any reward, so we focus on these nodes and ignore all other nodes.
A node $v$ that receives information from his parent observes a chain of identities $p_1, p_2, \ldots, p_h$, ending with its own identity ($p_h=v$) which follows its parent identity ($p_{h-1}$ is $v$'s parent if $v$ is not a seed. If $v$ is a seed the chain has only one element $p_1=v$). This chain is the result of the strategies of $v$'s ancestors, which we discuss below.
The chain $p_1, p_2, \ldots, p_h=v$ which $v$ observes defines $v$'s {\em level $l(v)$}, which is the maximal number of identities (true and fake) it can have in any completion of the observed chain to a chain of length $\mathcal H$. Thus, for $i\leq \mathcal H$ the level of $v$ is $l(v) = \mathcal H-i+1$ (note that the level of each seed is defined to be $\mathcal H$). For $i>\mathcal H$ we define $l(v)=0$.

A strategy of a node can only depends on the length of the chain it observes.
Note that for any chain of length at most $\mathcal H$ there is a one to one
mapping between the length of the chain $v$ observes and the level of $v$, so
we think of $v$'s strategy as a function of its level $l(v)$.

When a node $v$ receives the transaction, it decides whether it wants to send
the transaction onward, as well as the number of times it will duplicate
itself before sending. This decision is made separately with respect to each
child. One additional decision is whether the node wants to duplicate itself
when authorizing the transaction. As discussed above, both decisions are
based on the level of $v$ and nothing else. For each node $v$ and level $l$,
$c^{l,v}_1$ denotes the number of times $v$ clones itself before sending the
transaction to its $i$'th child. Additionally, $p^{l,v}$ denotes the number
of times a node duplicates itself before it tries to authorize the
transaction (note that if it succeeds it will use its last identity to report
its success). With these notations we can finally represent the strategies of
the nodes. The strategy $S_v$ of a node $v$ is represented as follows: for
every $1\leq l\leq H$, there is a tuple $S_v(l)=(c^{l,v}_1,\ldots,
c^{l,v}_d)$ where $0 \leq c^{l,v}_i \leq l-1$ and a number $0 \leq p^{l,v}
\leq l-1$. Observe that once we fix the strategies of all nodes, the chain
each node $v$ observes (if at all) is fixed.
% Note that this recursive definition might end up with negative!
%In particular, the level of a node $v$ that is the $i$'th child of $u$ in the tree can be computed recursively: if $l(u)=0$ then $l(v)=0$, and if $l(u)>0$ then $l(v)=l(u)-1-c^{l(u),u}_i$ (note that $c^{l(u),u}_i$  is the $i$'s element of $S_u(l)$).

If node $v$ of level $l$ decides not to send the transaction to its $i$'s child it can implement this by setting $c^{l,v}_i=l-1$. Now we can say that a node $v$ \emph{propagates information and does not duplicate at level $l$} if $S_v(l)=(0,\ldots, 0)$ and $p^{l,v}=0$. We also say that a node $v$ \emph{never duplicates and always propagates information}, if for every level $l$ node $v$ propagates information and does not duplicate at level $l$.
% Observe that if $v$ is aware of the transaction, then $H-l(v)+1$ is the length of the chain from the seed to $v$ (including clones).

Given a strategy profile $S$ and a reward scheme the utility of each node is defined as follows. An authorizer $w$ is chosen uniformly at random among the players that are aware of the transaction and try to authorize it.
Let $p_1, p_2, \ldots, p_h$ be the authorizing chain, it ends with the last identity of $w$ and starts with $p_1=s$ where $s$ is the seed that roots the tree of $w$. Denote $l = H-h+1$.
Then, we allocate rewards to all identities in the chain: $w$ gets a reward of $r^s_{1,l}$, its predecessor gets $r^s_{2,l}$ and so on. The total reward a node receives is the sum of rewards of its true and fake identities in the successful chain.

\subsubsection{The Proof Framework}

We start with the following definition. The intuition is that in a
$\varphi$-subgame a node that has $\mathcal H-l$ ancestors (including clones), does
not clone itself and propagates information if $l\leq \varphi+1$. Otherwise
for every child it duplicates itself at most $l-\varphi-1$ times.

\begin{definition}
The \emph{$\varphi$-subgame} is the $(\beta,\mathcal H)$-game with restricted strategy spaces: the strategy space of node $v$ includes only strategies such that for every remaining strategy profile of the rest of the players, every $i$ and $l(v)$: $c^{l(v),v}_i \leq \max \{ l(v) -\varphi-1, 0 \}$.
\end{definition}

The technical core of the proof is the following lemma:

\begin{lemma}\label{lemma-main}
In the $\varphi$-subgame, suppose that there are at least $7$ seeds, and at least $2 \beta^{-1} + 6$ additional nodes are aware of the transaction. Then, any strategy $(c^{l(v),v}_1,\ldots, c^{l(v),v}_d)$ of node $v$ such that some $c^{l(v),v}_i=l(v) -\varphi-1$ is dominated by the strategy $(c^{l(v)}_1,\ldots, c^{l(v),v}_{i-1},c^{l(v),v}_i-1,c^{l(v),v}_{i+1},\ldots ,c^{l(v),v}_d)$.
\end{lemma}

The proof of the lemma is in Subsection \ref{subsec-lemma-main}. Let us show
how to use the lemma in order to prove the theorem. First, we observe that
the $0$-subgame is simply the $(\beta,\mathcal H)$-game. If the conditions of the lemma
hold, we can apply the lemma to eliminate all strategies
$(c^{l(v),v}_1,\ldots, c^{l(v),v}_d)$ such that some $c^{l(v),v}_i=l(v) -1$.
Notice that now we have a $1$-subgame. Applying the lemma again we get a
$2$-subgame. Similarly, we repeatedly apply the lemma until we get a
$(\mathcal H-1)$-subgame. Notice that the only surviving strategy of each
node is $(0,\ldots, 0)$ for every $l(v)$: that is each node propagates
information and does not duplicate. This shows that there exists an
elimination order in which every node propagates information and does not
duplicate. Next, we prove that the strategy profile in which all nodes
propagate information and do not duplicate survives \emph{every} order of
iterated elimination of dominated strategies process. % (proof in the appendix):

\begin{lemma}\label{lemma-uniquness}
Let $T$ be a sub-game that is reached via iterated elimination of dominated
strategies in the $(\beta, \mathcal H)$-game, and suppose that there are  at least $7$ seeds, and
at least $2 \beta^{-1} + 6$ additional nodes are aware of the transaction. Then there
exists a strategy profile $s \in T$ in which every node at depth at most $\mathcal H$ fully propagates and do not
duplicate.
\end{lemma}
\begin{proof}
Let us assume that some elimination order ends with a sub-game that does not
contain any profile with full propagation and no duplication. Let $T'$ be the
last sub-game in the elimination order for which there is still a profile
with full propagation and no duplication for every node at depth at most
$\mathcal H$. It must be that for some player $i$ in $T'$, the strategy $s_i$
of full propagation and no duplication is dominated by another strategy
$s_i'$ in which this player either duplicates or does not fully distribute.
In particular, let us fix the behavior of the other players to be the profile
$s_{-i}$ in which they fully distribute and do not duplicate (such a profile
exists in $T'$ by assumption). For $s'_i$ to dominate $s_i$, it must be that
$u_i(s_i,s_{-i}) \leq u_i(s_i',s_{-i})$.

We will show however that the opposite holds, and thereby reach a
contradiction. We define a series of strategies $s_i^1, s_i^2, s_i^3, \ldots
,s_i^m$ such that $s_i^1 = s_i'$, and $s_i^m = s_i$ for which we shall show:

$u_i(s_i',s_{-i}) = u_i(s_i^1,s_{-i})< u_i(s_i^2,s_{-i})<\ldots <
u_i(s_i^m,s_{-i}) = u_i(s_i,s_{-i})$. Note, that the strategies $s_i^2,
s_i^3, \ldots ,s_i^{m-1}$ may or may not be in $T'$, and that we are merely
using them to establish the utility for $s_i, s_i'$.

The strategy $s_i^{j+1}$ is simply the strategy $s_i^j$, with one change:
player $i$ replicates itself one less time to one of its children.
Specifically, if player $i$ replicates itself $(c_1,\ldots ,c_d)$ times
correspondingly for each of its $d$ children, let $\nu = \arg\max_j (c_j)$ it
will instead replicate itself $(c_1,\ldots,c_\nu-1,\ldots ,c_d)$ times in the
new strategy.

We establish the fact that $u_i(s_i^j,s_{-i})< u_i(s_i^{j+1},s_{-i})$ by a
direct application of Claim~\ref{claim_internal}. The claim shows that
replicating one less time is better given sufficient external computation
power, and given that the node's descendants fully distribute with no
replication. Both of these are guaranteed in the profile $s'$.
\end{proof}

\subsubsection{Proof of Lemma \ref{lemma-main}}\label{subsec-lemma-main}
Now, let us observe a node $v$ and fix $l=l(v)$. The non-trivial case is when
$l\geq \varphi +1$. For convenience, we drop the superscript $(l(v),v)$ and
denote a strategy by ($c_1,\ldots, c_d)$. Without loss of generality we
assume $c_1 \leq c_2 \leq \ldots \leq c_d$. It will also be useful to define
$y_i=l-c_i-1$. Note that our previous assumption implies that $y_1 \geq y_2
\geq \ldots \geq y_d$.

Let $s_{-v}$ be a strategy profile of all other nodes except $v$, and denote
by $A_{s_{-v}}(y_i)$ the size of the subtree rooted at the $i$'th child of
$v$ that is aware of the transaction. The utility of $v$ is then:
\begin{equation*}
U_v^{s_{-v}}(y_1, \ldots, y_d) = \underbrace{\frac{1 + \beta \cdot l}{k + \sum_{i=1}^d A_{s_{-v}}(y_i)}}_{v \hbox{ authorizes}} +\underbrace{\frac{ \sum_{i=1}^d \beta (l-y_i)\cdot A_{s_{-v}}(y_i)}{k + \sum_{i=1}^d A_{s_{-v}}(y_i)}}_{\hbox{a decedent of $v$ authorizes}}
\end{equation*}
where $k$ is the number of nodes which are aware of the transaction and are not descendants of $v$ (this number includes $v$ itself). We
show that $(c_1,\ldots, c_d)$, where $c_d=l(v) -\varphi-1$ is dominated by
$(c_1,\ldots ,c_d-1)$. For this purpose, it is sufficient to show the
following claim:

\begin{claim} \label{claim_internal}
Let $T_{-v}$ be the set of strategy profiles of all nodes other than $v$ in
which:
\begin{enumerate}
\item The subtree rooted at $v$'s $d$'th child spans a complete $d$-ary tree of
   height $y_d$ when $v$ duplicates itself $c_d$ times, or of a
   height $y_d+1$ when $v$ duplicates itself $c_d-1$ times.

\item There are at least $k \geq 2\beta^{-1}+6d \frac{d^{y_d}-1}{d-1} +6$
    nodes that are not descendants of $v$ that are aware of the
    transaction.
\end{enumerate}
Then: $ \displaystyle \forall s_{-v} \in T_{-v} \quad  U_v^{s_{-v}}
(c_1,\ldots ,c_d-1) > U_v^{s_{-v}}(c_1,\ldots ,c_d) $
\end{claim}

\begin{figure}[!ht]
  \begin{center}
\includegraphics[width=0.7\linewidth]{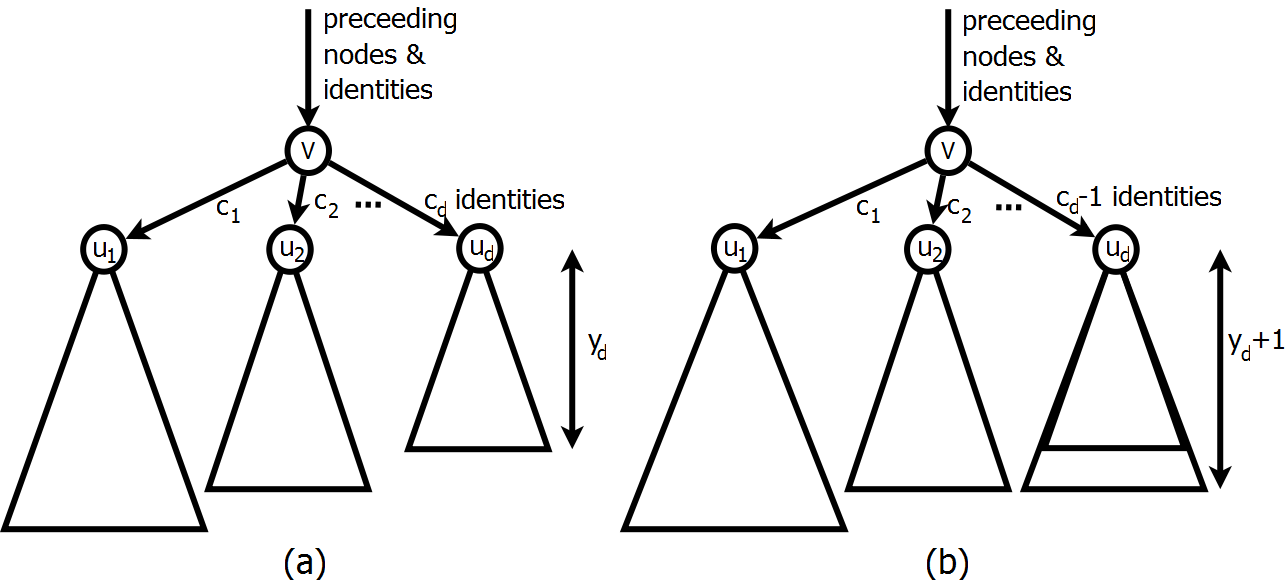}
  \end{center}
  \caption{\label{fig_trees} Node $v$ distributes the transaction to it's $i$'th child using $c_i$ identities.
  It considers distributing the transaction to the $d$'th child with one less identity, which enables the distribution tree
  below that child to extend $y_d+1$ levels (assuming all nodes below use only one identity).}
\end{figure}

The two scenarios are depicted in Figure~\ref{fig_trees}. Notice, that in any
$\varphi$-subgame, both of the conditions for the lemma hold in every
profile: in a $\varphi$-subgame we have that if node $v$ clones itself at
most $l(v) -\varphi-2$ times then his child will span a tree that contains a
full $d$-ary tree of height $\varphi$, and we are guaranteed $k$ nodes that
are aware of the transaction. We therefore turn to prove the claim.

\begin{proof}
The new strategy $(c_1,\ldots ,c_d-1)$ in effect distributes the transaction
to one additional layer of the subtree rooted at the $d$'th child of $v$.
That extra layers contains $d^{y_d}$ nodes. The utility of this strategy is
therefore:

\begin{equation*}
U_v^{s_{-v}}(y_1, \ldots,y_{d-1}, y_d+1) = \frac{1 + \beta \cdot l + \sum_{i=1}^{d-1} \beta (l-y_i)\cdot A_{s_{-v}}(y_i) + \beta(l-y_d-1)(A_{s_{-v}}(y_d)+d^{y_d})}{k+ \sum_{i=1}^d A_{s_{-v}}(y_i) +d^{y_d}}
\end{equation*}

So we have to show that:
\begin{equation*}
\forall s_{-v} \quad  U_v^{s_{-v}}(y_1, \ldots,y_{d-1}, y_d+1) > U_v^{s_{-v}}(y_1, \ldots, y_d)
\end{equation*}

For convenience, we will drop the index $s_{-v}$, but remember that we must
check for every possible strategy profile of the other nodes that:
\begin{equation*}
\frac{1 + \beta \cdot l + \sum_{i=1}^{d-1} \beta (l-y_i)\cdot A(y_i) + \beta(l-y_d-1)(A(y_d)+d^{y_d})}{k+ \sum_{i=1}^d A(y_i) +d^{y_d}} >
\frac{1 + \beta \cdot l + \sum_{i=1}^d \beta (l-y_i)\cdot A(y_i)}{k+ \sum_{i=1}^d A(y_i)}
\end{equation*}
multiplying by the denominator and dividing both sides by $\beta$:
\begin{equation*}
\left(k+ \sum_{i=1}^d A(y_i)\right)\cdot (l-y_d-1)(A(y_d)+d^{y_d})>
\end{equation*}
\begin{equation*}
\left(k+ \sum_{i=1}^d A(y_i)\right)(l-y_d)\cdot A(y_d) + d^{y_d} \cdot \left(\beta^{-1}+ l + \sum_{i=1}^d (l-y_i)\cdot A(y_i) \right)
\end{equation*}

\begin{equation*}
\left(k+ \sum_{i=1}^d A(y_i)\right)\cdot (l-y_d-1) \cdot d^{y_d} - \left(k+ \sum_{i=1}^d A(y_i)\right) \cdot A(y_d)>
d^{y_d} \cdot \left(\beta^{-1}+ l + \sum_{i=1}^d (l-y_i)\cdot A(y_i) \right)
\end{equation*}

\begin{equation*}
\left(k+ \sum_{i=1}^d A(y_i)\right) \cdot \left( l-y_d-1-\frac{A(y_d)}{d^{y_d}} \right) >
\left(\beta^{-1}+ l + \sum_{i=1}^d (l-y_i)\cdot A(y_i) \right)
\end{equation*}

Note that $l-y_d-1-\frac{A(y_d)}{d^{y_d}}>0$ since $l-y_d\geq 2$ and $\frac{A(y_d)}{d^{y_d}}<1$. Therefore we can divide by $l-y_d-1-\frac{A(y_d)}{d^{y_d}}$, after some rearranging we get that:

\begin{equation*}
k > \frac{\beta^{-1}+ l + \sum_{i=1}^d (l-y_i)\cdot A(y_i)}{l-y_d-1-\frac{A(y_d)}{d^{y_d}}} -\sum_{i=1}^d A(y_i)
\end{equation*}

\begin{equation*}
k > \frac{\beta^{-1} + l + \sum_{i=1}^d (l-y_i -l+y_d+1+\frac{A(y_d)}{d^{y_d}})\cdot A(y_i) }{ l-y_d-1-\frac{A(y_d)}{d^{y_d}}}
\end{equation*}

\begin{equation*} \label{eqn_toBound}
k\cdot(l-y_d-1-\frac{A(y_d)}{d^{y_d}}) > \beta^{-1}+ l + \sum_{i=1}^d (-y_i+y_d+1+\frac{A(y_d)}{d^{y_d}})\cdot A(y_i) = (*)
\end{equation*}

Recall that that the $d$'s child spans a full tree of height $y_d$: $A(y_d) = \frac{d^{y_d}-1}{d-1}$. We bound the right hand side of Equation~\ref{eqn_toBound} as follows:

\begin{align*}
(*) \leq  \beta^{-1}+l
&+ \displaystyle \sum_{i:y_i=y_d}^d \left(1+\frac{A(y_d)}{d^{y_d}}\right)\cdot A(y_i)
+ \displaystyle \sum_{i:y_i=y_d+1}^d \left(\frac{A(y_d)}{d^{y_d}}\right)\cdot A(y_i)\\
&+ \displaystyle \sum_{i:y_i>y_d+1}^d \left(-1+\frac{A(y_d)}{d^{y_d}}\right)\cdot A(y_i)\end{align*}

The last summation is negative because $\frac{A(y_d)}{d^{y_d}}<1$, and once we substitute $A(y_d)$ into the expression we have that:

\begin{equation*}
(*) \leq \beta^{-1}+l
+ \displaystyle \sum_{i:y_i=y_d}^d \left(1+\frac{A(y_d)}{d^{y_d}}\right)\cdot A(y_d)
+ \displaystyle \sum_{i:y_i=y_d+1}^d \frac{A(y_d)}{d^{y_d}}\cdot A(y_d+1)
\end{equation*}

Now notice that:

\begin{equation*}
\left( 1+\frac{A(y_d)}{d^{y_d}} \right)\cdot A(y_d) = \frac{A(y_d)+d^{y_d}}{d^{y_d}} \cdot A(y_d) = \frac{A(y_d+1)\cdot A(y_d)}{d^{y_d}}
\end{equation*}

and so we can write:
\begin{align}
(*) \leq \beta^{-1}+l + \displaystyle\sum_{i: y_i=y_d \lor y_i = y_d+1} (1+\frac{A(y_d)}{d^{y_d}})\cdot A(y_d) &\leq
\beta^{-1}+l + \displaystyle\sum_{i: y_i=y_d \lor y_i = y_d+1}  2\cdot A(y_d) \\
&\leq \beta^{-1}+l + 2\cdot d \cdot A(y_d)
\label{eqn_cases}
\end{align}

Therefore, we are looking for $k$ such that:
\begin{equation*}
k > \frac{ \beta^{-1}+l + 2\cdot d \cdot A(y_d)}{l-y_d-1-\frac{A(y_d)}{d^{y_d}}}=(**)
\end{equation*}

We now divide into two cases.
\paragraph{Case I:}{ $l \geq 2y_d +4$}. In this case
By using the fact $\frac{A(y_d)}{d^{y_d}}<1$ and continuing from
Equation~(\ref{eqn_cases}) we can write:

\begin{align*}
(**) &< \frac{\beta^{-1}+l +  2\cdot d \cdot A(y_d) }{ l-y_d-2} \\&\leq  \frac{\beta^{-1}+l +  2\cdot d \cdot A(y_d) }{ \frac{l}{2}} \\&\leq
\frac{2 \cdot \beta^{-1}}{l} +2 +\frac{4\cdot d \cdot A(y_d)}{l} \\&\leq \frac{1}{2} \cdot \beta^{-1} +2 + d \cdot A(y_d)
\end{align*}
where the last transition relies on the fact that $l \geq 4$.

\paragraph{Case II:}{ $l \leq 2y_d +3$}. Notice that by definition $l$ is always at least $y_d+2$. Therefore, by continuing from Equation~(\ref{eqn_cases}) we have:

\begin{equation*}
(**) \leq \frac{\beta^{-1}+l + 2\cdot d \cdot A(y_d) }{2-1-\frac{A(y_d)}{d^{y_d}}} \leq
\frac{\beta^{-1}+l + 2\cdot d \cdot A(y_d) }{1-\frac{1}{d-1}} =
\frac{d-1}{d-2} \left( \beta^{-1}+l + 2\cdot d \cdot A(y_d) \right)
\end{equation*}
where for the second transition we used:

\begin{equation*}
\frac{A(y_d)}{d^{y_d}} = \frac{d^{y_d}-1}{(d-1)d^{y_d}} < \frac{1}{d-1}
\end{equation*}
Notice that $A(y_d) \geq y_d $, hence, $l\leq 2A(y_d)+3$. Recall that $d\geq 3$ we can continue to bound $(**)$:

\begin{equation*}
(**)\leq  \frac{d-1}{d-2} \left( \beta^{-1} + (2 \cdot d+2)\cdot A(y_d)  +3 \right) \leq 2 \cdot \beta^{-1} + 6 \cdot d\cdot A(y_d)  +6
\end{equation*}
From Case I, and case II combined, we have that:

\begin{equation*}
(**)\leq  2 \beta^{-1} + 6 d\cdot A(y_d)  +6
\end{equation*}
\end{proof}

Thus, if $k\geq 2 \beta^{-1} + 6 d\cdot A(y_d)  +6$, for any node $v$ the
strategy of choosing $c_i=l(v) -\varphi-1$ for some child $i$ is dominated by
the strategy $(c_1,\ldots, c_{i-1},c_i-1,c_{i+1},\ldots ,c_d)$. Since $k$ is
the number of nodes outside the tree rooted by $v$ then for the lemma to hold
it is sufficient to have 7 seeds (each with $d$ children that span trees at
height $y_d$), and $2\beta^{-1}+6$ additional nodes.

%% file: impossibility.tex
\section{Impossibility Result for Dominant Strategy Mechanisms}

In the previous sections, we presented several reward schemes and analyzed
their behavior in equilibrium. The solution concept we analyzed was iterated
removal of dominated strategies. Although iterated removal of dominated
strategies is a strong solution concept, a natural goal is to seek reward
schemes that use even stronger solution concepts. We now show that there are
no dominant-strategy reward schemes.

\begin{theorem}
Suppose that $ H\geq 3$. There is no Sybil-proof reward scheme in which information propagation and no duplication are dominant strategy for all nodes at depth $3$ or less.
\end{theorem}
\begin{proof}
Consider a scenario in which all nodes except one seed $s$ do not propagate information. Suppose that all the direct children of the seed $s$ play the following strategy:
\begin{itemize}
\item If they have one predecessor then they do not propagate information. In case they authorize the transaction they pretend that they are a level $3$ node, by duplicating themselves once (so if they authorize the transaction the resulting chain has a length of $3$).
\item If they have more than one predecessor then they still do not propagate information. However, if they authorize the transaction they do not duplicate themselves (so again if they authorize the transaction the resulting chain has a length of $3$).
\end{itemize}
All other nodes, including the children of the children of $s$, do not
propagate information and do not duplicate themselves.

Denote by $d$ the number of direct children of $s$, and consider the possible
strategies of $s$. The seed $s$ can propagate information to all its
children. In this case his children will not propagate information and one of
them authorizes the transaction it pretends that it is a level $3$ nodes.
Using the fact that there are $t$ seeds that are aware of the transaction,
the utility of $s$ in this case is:
\begin{equation*}
\frac {r_{1,1}+ d\cdot r_{3,3}}{t+d}
\end{equation*}
However, suppose that $s$ duplicate itself once and propagate information to
its children. In this case again its $d$ children will be aware of the
transaction. But now the utility of $s$ in this case is:
\begin{equation*}
\frac {r_{1,1} + d\cdot (r_{2,3}+r_{3,3})}{t+d}
\end{equation*}
Since information propagation and no duplication should be a dominant
strategy for $s$, we must have:
\begin{eqnarray*}
\frac {r_{1,1}+ d\cdot r_{3,3}}{t+d}&\geq& \frac {r_{1,1} + d\cdot (r_{2,3}+r_{3,3})}{t+d} \\
{r_{1,1}+ d\cdot r_{3,3}}&\geq& r_{1,1} + d\cdot (r_{2,3}+r_{3,3})\\
0&\geq& r_{2,3}
\end{eqnarray*}
The contradiction will be reached by proving the following lemma:
\begin{lemma}
Suppose that $ H\geq 3$. In every Sybil-proof reward scheme in which
information propagation and no duplication are dominant strategy for all
nodes at depth $3$ or less, we have that $r_{2,3}>0$.
\end{lemma}
\begin{proof}
Fix some node $u$ that receives the transaction with exactly one predecessor. Denote by $k$ the number of nodes that are aware of the transaction and are not descendants of $u$ (this number includes $u$). Since $u$'s dominant strategy is to propagate information to its $d$ children without duplication, the utility of this strategy is no worse than the utility of not propagating information at all. Suppose that the children never propagate information and never duplicate themselves:
\begin{eqnarray*}
\frac {r_{1,2}}{k}&\leq& \frac {r_{1,2} + d\cdot r_{2,3}}{k+d} \\
\end{eqnarray*}
Since we assume that the reward for the authorizer is strictly positive, we have that $r_{1,2}>0$ and therefore:
\begin{eqnarray*}
\frac {r_{1,2}}{k+d}&<& \frac {r_{1,2} + d\cdot r_{2,3}}{k+d} \\
\end{eqnarray*}
Which gives us that $r_{2,3}>0$, as needed.
\end{proof}

\end{proof}

%% file: overview.tex
\section{A Brief Overview of Bitcoin} \label{sec-bitcoin}

In this section we give a brief overview of the Bitcoin protocol. This is by
no means a complete description. The main purpose of this section is to help
understanding our modeling choices, and why our proposed reward schemes can
be implemented within the context of the existing protocol. The reader is
referred to \cite{Satoshi} and to the Bitcoin wiki for a complete description.

\subsection{Signing Transactions and Public Key Cryptography}
The basic setup of electronic transactions relies on public key cryptography.
When Alice wants to transfer $50$ coins to Bob, she signs a transaction using
her private key. Hence, everyone can verify that Alice herself initiated this
transaction (and not someone else). Bob, in turn, is identified as the target
of the transfer using his public key. For the money to be actually
transferred from Alice's account to Bob's account, some entity has to keep
track of the last owner of the coins, and to mark Bob as the new owner. Otherwise, Alice could ``double spend'' her money -- first transfer
the coins to Bob, then transfer the same coins again to Charlie.
Traditionally, this role was fulfilled by banks. In return, banks tended to
charge high fees, for example in international transfers.

\subsection{Agreeing on the History by Majority of Processing Power}
Bitcoin suggests a different solution to this problem. A peer to peer network
is used to validate all transactions. Nodes in the network agree on a common
history using a ``majority of processing-power mechanism''. A mechanism can
not rely on a numerical majority of the nodes, as it is relatively easy to
create additional identities in a network, for example by spoofing IP
addresses. We first describe the protocol that the network implements and
then explain why the history is accepted only if it is agreed upon by nodes
that together control a majority of the processing power.

Let us assume again that Alice wants to transfer $50$ coins to Bob. Alice
will send her signed transaction to some of the nodes in the network. Next,
these nodes will forward the transaction to all of their neighbors in the
network and so on. A node that receives the transaction first verifies that
this is a valid transaction (e.g. that the money being transferred indeed
belongs to Alice). If successful, the node adds this transaction to the block
of transactions it attempts to authorize (a block is simply a set of transactions; the specific details on the
structure of this block are omitted). To authorize a block, a node has to
solve an inverse Hash problem. More specifically, its goal is to add some
bits (nonce) to the block such that the Hash value of the new block begins
with some predefined number of zeroes.

The number of zeroes is adjusted such that the average time it takes the
network to authorize a block is fixed. The protocol uses a clever method to
aggregate transactions (a Merkle tree) which assures that the size of the
string to be hashed is also fixed. Additional information that is included in
the block is the hash of the previously authorized block. So, in fact, the
authorized blocks form a chain in which each block identifies the one the
preceded it.

When a node authorizes a block it broadcasts to the network the new block and
the proof of work (the string which is added to the block to get the desired
hash). If a node receives two different authorized blocks it adopts the one
which is part of the longer chain.

Let us argue briefly and informally why the history is determined by the
majority of the processing power (\cite{Satoshi}). That is, the probability
that a group of malicious nodes manages to change the history decreases as
the fraction of processing power they control decreases. Assume that a group
of malicious nodes, that does not have the majority of the processing power,
wants to change the chain that has currently been adopted by the other nodes.
Since the malicious nodes own less processing power, their authorization rate
is slower than the authorization rate of the majority of the network, which
is honest. Therefore, the malicious nodes will not be able to produce a
longer chain than the one that is currently accepted by the network (with
high probability). The longer chain is the one that will be adopted by the
honest nodes in the network (and will also grow longer as they continue to
extend it). In fact, once a block has been accepted into the chain, the
probability that a longer chain that displaces it will appear decreases
exponentially with time (as the chain that follows it grows longer it is
harder to manufacture a longer alternative one). See~\cite{Satoshi} for a
more formal argument.

\subsection{Transaction Fees}
To incentivize nodes to participate in the peer to peer network and invest
effort in authorizing transactions, rewards have to be allocated. Currently,
in the initial stages of the protocol this is done by giving the authorizer
of a block a fixed amount of bitcoins (this also the only method for printing
new bitcoins). This fixed amount will be reduced every few years at a rate
determined by the protocol. As the money creation is slowly phased out, there
will be a need to reward the nodes differently. The protocol has been
designed with this in mind. It already requires the transaction initiator to
specify a fee for authorizing her transaction. As we explained in the
introduction this introduces an incentive problem since nodes prefer to keep
the transaction to themselves instead of broadcasting it.

\subsection{On Implementing the Hybrid Scheme}\label{subsec-implementation}

We start with a rough description of how fees for authorizing transactions in
Bitcoin are currently implemented. Alice produces a transaction record in
which she transfers some of her coins to Bob. She specifies the fee (if any)
for authorizing this transaction in a field in the transaction record, and
then cryptographically signs this record. When Victor authorizes a
transaction, the implication is that Alice transferred some amount of money
to Bob, and some amount of money (specified in the fee field) to Victor.

Any $(\beta,\mathcal H)$-almost uniform scheme can be implemented similarly
(and thus, the Hybrid scheme can be implemented similarly). Alice produces a
different transaction record for every seed, in which she specifies $\beta$,
$\mathcal H$, and some amount of coins $f$ (we normalized $f=1$ in the
previous sections -- the total fee if a node in tree rooted by the seed
authorizes the transaction will be $(1+\beta\cdot \mathcal H)\cdot f$). We
modify the protocol so that every node that transfers the transaction to its
children cryptographically signs the transaction, and specifically identifies
the child to whom the information is being sent using that child's public
key. The transaction information, therefore includes a ``chain of custody''
for the transaction which will be different for every node that tries to
authorize the transaction. The implication of Victor authorizing a
transaction is that Alice transferred some amount of money to Bob, a fee of
$f\cdot \beta$ to every node in the path (as it appears in the authorized
transaction), and a fee of $f+\beta \cdot (\mathcal H-h+1)\cdot f$ to Victor if
the path was of length $h\leq \mathcal H$. Payments will not be awarded if
the ``chain of custody'' is not a valid chain that leads from Alice to
Victor.